\documentclass[acmtalg]{acmsmall}
\usepackage[latin1]{inputenc}
\usepackage{latexsym}
\usepackage{amsmath}
\usepackage{amssymb}
\usepackage{graphicx}
\usepackage{ifthen}
\usepackage{calc}

 \newtheorem{defin}[theorem]{Definition}

\newcommand{\nc}{\newcommand}
\nc{\rnc}{\renewcommand}
\nc{\nev}{\newenvironment}
\renewcommand{\fboxsep}{2mm}

\nc{\bigmid}{\;\big|\;}
\nc{\Bigmid}{\;\Big|\;}
\rnc{\max}{\textup{max}}
\rnc{\min}{\textup{min}}
\rnc{\log}{\textup{log}\;}

\rnc{\fboxsep}{2mm}
\newsavebox{\fminibox}
\newlength{\fminilength}
\newenvironment{fminipage}[1][\linewidth]{
  \setlength{\fminilength}{#1-2\fboxsep-2\fboxrule}%
  \begin{lrbox}{\fminibox}\begin{minipage}{\fminilength}}{
    \end{minipage}\end{lrbox}\noindent\fbox{\usebox{\fminibox}}
}

\newlength{\pwidth}

\newenvironment{problem}[3]{
  \begin{center}
    \ifthenelse{\equal{#1}{}}{\setlength{\pwidth}{\columnwidth}}%
      {\setlength{\pwidth}{#1cm}}%
      \nc{\instance}{\item[Instance]}%
      \nc{\parameter}{\item[Parameter]}%
      \rnc{\problem}{\item[Problem]}%
      \begin{fminipage}[\pwidth]\upshape
        \ifthenelse{\equal{#3}{}}{}{{\scshape #3}}
        \begin{list}{}{
            \ifthenelse{\equal{#2}{}}%
            {\settowidth{\labelwidth}{\textit{Parameter:}}}%
          {\settowidth{\labelwidth}{\textit{#2:}}}
          
          \setlength{\leftmargin}{\labelwidth+\labelsep}
          \setlength{\itemsep}{0ex}
          \setlength{\parsep}{0ex}
        }
      }{
      \end{list}
    \end{fminipage}
  \end{center}
}

\newcommand{\nprob}[4][]{
  \index{\textsc{#2}}
  \begin{problem}{#1}{Instance}{#2}
    \instance #3.
    \problem #4.
  \end{problem}
}

\rnc{\mod}[1]{\,(\textup{mod }#1)}

\rnc{\mathbf}[1]{\text{\itshape\bfseries #1}}

\nc{\PTIME}{\textup{PTIME}}
\nc{\NP}{\textup{NP}}
\nc{\FPT}{\textup{FPT}}
\nc{\W}[1]{\textup{W[#1]}}

\nc{\CSP}{\textsc{Csp}}
\nc{\pCSP}{\textit{p-}\textsc{Csp}}
\nc{\HOM}{\textsc{Hom}}
\nc{\pHOM}{\textit{p-}\textsc{Hom}}
\nc{\EMB}{\textsc{Emb}}
\nc{\pEMB}{\textit{p-}\textsc{Emb}}
\nc{\tw}{\textup{tw}}
\nc{\hw}{\textup{hw}}
\nc{\ghw}{\textup{ghw}}
\nc{\mw}{\textup{mw}}
\nc{\fhw}{\textup{fhw}}
\nc{\aw}{\textup{aw}}
\nc{\RA}{\textup{RA}}
\nc{\weight}{\textup{weight}}

\acmVolume{0}
\acmNumber{0}
\acmArticle{0}
\acmYear{0}
\acmMonth{0}

 
\begin{document}
\title{
Constraint Solving via Fractional Edge Covers\footnote{An extended abstract of the paper appeared in the 
Proceedings of the seventeenth annual ACM-SIAM Symposium on Discrete Algorithms (SODA 2006).}}
\author{Martin Grohe\affil{RWTH Aachen University, 
Lehrstuhl für Informatik 7, 
Aachen, 
Germany, 
    \texttt{grohe@informatik.rwth-aachen.de}}
D\'aniel Marx\affil{
Computer and Automation Research Institute,
Hungarian Academy of Sciences (MTA SZTAKI),
Budapest, Hungary,
\texttt{dmarx@cs.bme.hu.} \footnote{Research  supported by the European Research Council (ERC)  grant 
``PARAMTIGHT: Parameterized complexity and the search for tight complexity results,'' reference 280152.}}}
\date{}
\begin{abstract}
  Many important combinatorial problems can be modeled as constraint
  satisfaction problems. Hence identifying polynomial-time solvable
  classes of constraint satisfaction problems has received a lot of
  attention. In this paper, we are interested in structural properties
  that can make the problem tractable. So far, the largest structural
  class that is known to be polynomial-time solvable is the class of
  bounded hypertree width instances introduced by \citeN{gotleosca02}. Here we identify a new class of
  polynomial-time solvable instances: those having bounded fractional
  edge cover number. 

  Combining hypertree width and fractional edge cover number, we then
  introduce the notion of fractional hypertree width. We prove that
  constraint satisfaction problems with bounded fractional hypertree
  width can be solved in polynomial time (provided that a the tree
  decomposition is given in the input). Together with a recent
  approximation algorithm for finding such decompositions
  \cite{marx-fractional-talg}, it follows that bounded fractional
  hypertree width is now the most general known structural property that
  guarantees polynomial-time solvability.
\end{abstract}
\maketitle

\section{Introduction}
Constraint satisfaction problems form a large class of combinatorial problems
that contains many important ``real-world'' problems. An instance of a
constraint satisfaction problem consists of a set $V$ of variables, a domain
$D$, and a set $C$ of constraints. For example, the domain may be $\{0,1\}$,
and the constraints may be the clauses of a 3-CNF-formula. The objective is to
assign values in $D$ to the variables in such a way that all constraints are
satisfied. In general, constraint satisfaction problems are NP-hard;
considerable efforts, both practical and theoretical, have been made to
identify tractable classes of constraint satisfaction problems.

On the theoretical side, there are two main directions towards identifying
polynomial-time solvable classes of constraint satisfaction
problems. One is to restrict the
\emph{constraint language}, that is, the type of constraints that are allowed
(see, for example,
\cite{bul02,DBLP:journals/tocl/Bulatov11,buljeakro01,fedvar98,jeacohgys97,schae78}). Formally, the
constraint language can be described as a set of relations on the domain. The
other direction is to restrict the \emph{structure} induced by the constraints
on the variables (see, for example,
\cite{DBLP:journals/jcss/CohenJG08,dalkolvar02,decpea89,fre90,kolvar98,groschweseg01,DBLP:journals/jacm/Grohe07}). The present work
goes into this direction; our main contribution is the identification of a
natural new class of structurally tractable constraint satisfaction problems.

The \emph{hypergraph} of an instance $(V,D,C)$ has $V$ as its vertex
set and for every constraint in $C$ a hyperedge that consists of all
variables occurring in the constraint. For a class $\mathcal H$ of
hypergraphs, we let $\CSP(\mathcal H)$ be the class of all instances
whose hypergraph is contained in $\mathcal H$. The central question
is for which classes $\mathcal H$ of hypergraphs the problem
$\CSP(\mathcal H)$ is tractable. Most recently, this question has been
studied in
\cite{chedal05,DBLP:journals/jcss/CohenJG08,marx-tocs-truthtable,marx-stoc2010-csp,ggmss05}.  It is worth pointing out that the corresponding
question for the graphs (instead of hypergraphs) of instances, in
which two variables are incident if they appear together in a
constraint, has been completely answered in
\cite{DBLP:journals/jacm/Grohe07,groschweseg01} (under the complexity
theoretic assumption $\FPT\neq\W1$): For a class $\mathcal G$ of
graphs, the corresponding problem $\CSP(\mathcal G)$ is in polynomial
time if and only if $\mathcal G$ has bounded tree width. This can be
generalized to $\CSP(\mathcal H)$ for classes $\mathcal H$ of
hypergraphs of \emph{bounded hyperedge size} (that is, classes
$\mathcal H$ for which $\max\{|e|\mid\exists\,H=(V,E)\in\mathcal
H:\;e\in E\}$ exists).  It follows easily from the results of
\cite{DBLP:journals/jacm/Grohe07,groschweseg01} that for all classes
$\mathcal H$ of bounded hyperedge size,
\begin{equation}
  \label{eq:tw}
  \CSP(\mathcal H)\in\PTIME\iff\mathcal H\text{ has bounded tree width}
\end{equation}
(under the
assumption $\FPT\neq\W1$).

It is known that (1) does not generalize to arbitrary classes $\mathcal H$ of
hypergraphs (we will give a very simple counterexample in
Section~\ref{sec:prels}). The largest known family of classes of hypergraphs
for which $\CSP(\mathcal H)$ is in \PTIME\ consists of all classes of bounded
\emph{hypertree width} \cite{gotleosca02,gotleosca03,DBLP:journals/ai/GottlobLS00}.  Hypertree
width is a hypergraph invariant that generalizes acyclicity
\cite{ber76,fag83,yan81}. It is a very robust invariant; up to a constant
factor it coincides with a number of other natural invariants that
measure the global connectivity of a hypergraph~\cite{DBLP:journals/ejc/AdlerGG07}. On classes
of bounded hyperedge size, bounded hypertree width coincides with bounded tree
width, but in general it does not. It has been asked in
\cite{chedal05,DBLP:journals/jcss/CohenJG08,ggmss05,DBLP:journals/jacm/Grohe07} whether there are classes $\mathcal
H$ of unbounded hypertree width such that $\CSP(\mathcal H)\in\PTIME$. 
We
give an affirmative answer to this question.

Our key result states that $\CSP(\mathcal H)\in\PTIME$ for all classes
$\mathcal H$ of bounded \emph{fractional edge cover number}. A
\emph{fractional edge cover} of a hypergraph $H=(V,E)$ is a mapping
$x:E\to[0,\infty)$ such that $\sum_{e\in E, v\in e}x(e)\ge 1$ for all
$v\in V$.  The number $\sum_{e\in E}x(e)$ is the \emph{weight} of $x$.
The \emph{fractional edge cover number} $\rho^*(H)$ of $H$ is the minimum of
the weights of all fractional edge covers of $H$. It follows from standard
linear programming results that this minimum exists and is
rational. Furthermore, it is
easy to construct classes $\mathcal H$ of hypergraphs that have bounded
fractional edge cover number and unbounded hypertree width (see
Example~\ref{exa:rho*<ghw}).

We then start a more systematic investigation of the interaction
between fractional covers and hypertree width. We propose a new
hypergraph invariant, the \emph{fractional hypertree width}, which
generalizes both the hypertree width and fractional edge cover number
in a natural way.  Fractional hypertree width is an interesting hybrid
of the ``continuous'' fractional edge cover number and the
``discrete'' hypertree width.  We show that it has properties that are
similar to the nice properties of hypertree width. In particular, we
give an approximative game characterization of fractional hypertree
width similar to the characterization of tree width by the ``cops and robber'' game \cite{seytho93}. Furthermore, we prove that for
classes $\mathcal H$ of bounded fractional hypertree width, the
problem $\CSP(\mathcal H)$ can be solved in polynomial time provided
that a fractional hypertree decomposition of the underlying hypergraph
is given together with the input instance.  We do not know if for
every fixed $k$ there is a polynomial-time algorithm for finding a
fractional hypertree decomposition of width $k$. However, a recent
result \cite{marx-fractional-talg} shows that we can find an
approximate decomposition whose width is bounded by a (cubic) function
of the fractional hypertree width. This is sufficient to show that
$\CSP(\mathcal H)$ is polynomial-time solvable for classes $\mathcal
H$ of bounded fractional hypertree width, even if no decomposition is
given in the input. Therefore, bounded fractional hypertree width is
the so far most general hypergraph property that makes $\CSP(\mathcal
H)$ polynomial-time solvable. Note that this property is strictly more
general than bounded hypertree width and bounded fractional edge cover
number (see Figure~\ref{fig:sets}).
\begin{figure}[t]
\begin{center}
\includegraphics[width=0.55\linewidth]{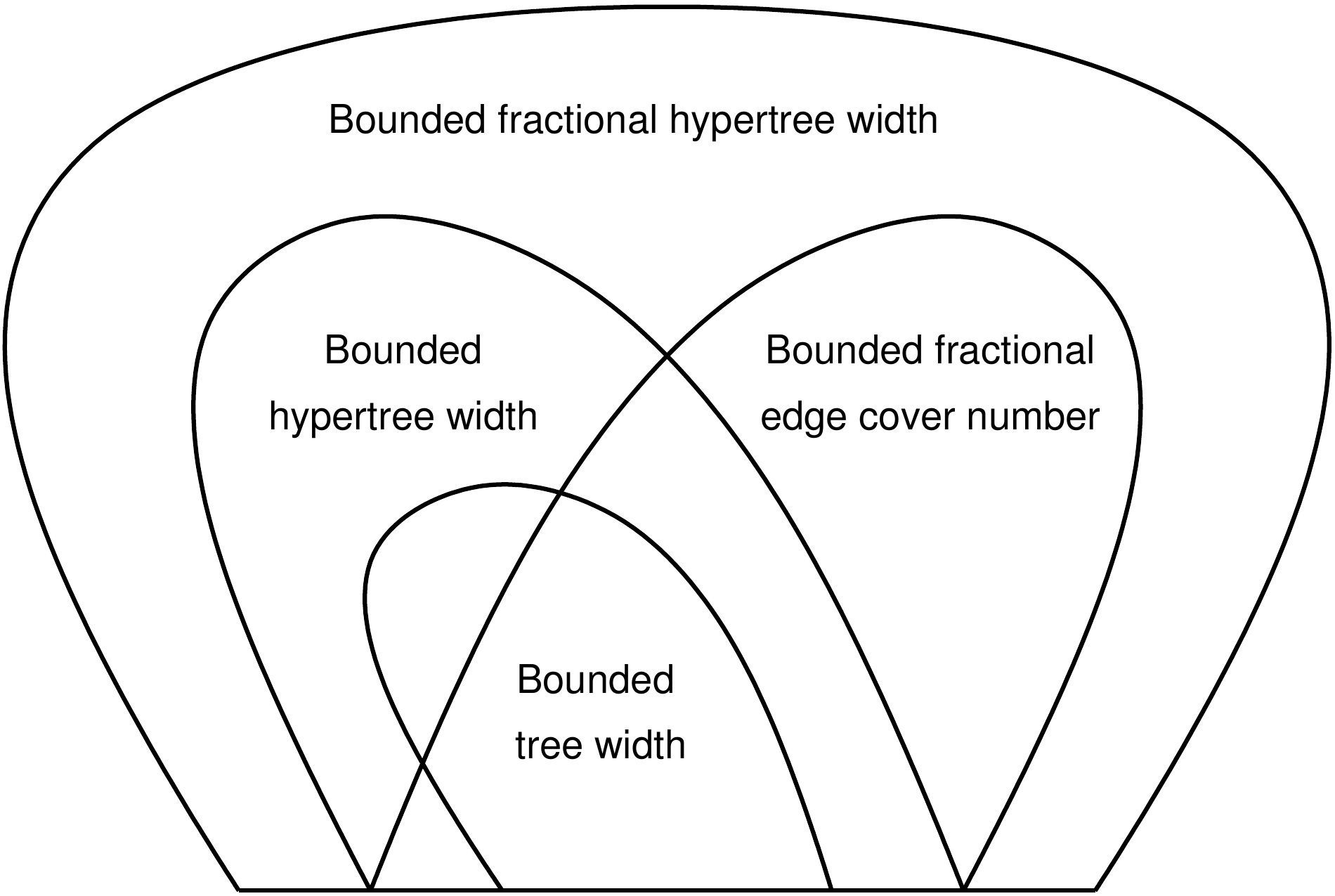}
\caption{Hypergraph properties that make CSP polynomial-time solvable.}\label{fig:sets}
\end{center}
\end{figure}

For classes $\mathcal H$ of hypergraphs of bounded fractional
hypertree width, we also show that all solutions of an instance of
$\CSP(\mathcal H)$ can be computed by a polynomial delay
algorithm. Closely related to the problem of computing all solutions
of a \CSP-instance is the problem of evaluating a conjunctive query
against a relational database. We show that conjunctive queries whose
underlying hypergraph is in $\mathcal H$ can be evaluated by a
polynomial delay algorithm as well.  Finally, we look at the
homomorphism problem and the embedding problem for relational
structures. The homomorphism problem is known to be equivalent to the
CSP \cite{fedvar98} and hence can be solved in polynomial time if the
left-hand side structure has an underlying hypergraph in a class
$\mathcal H$ of bounded fractional hypertree width. This implies that
the corresponding embedding problem, parameterized by the size of the
universe of the left-hand side structure, is fixed-parameter
tractable. Recall that a problem is {\em fixed-parameter tractable}
(FPT) by some parameter $k$ if the problem can be solved in time
$f(k)\cdot n^{O(1)}$ for a computable function $f$ depending only on
$k$. In particular, if a problem is polynomial-time solvable, then it
is fixed-parameter tractable (with any parameter $k$). 

\textbf{Follow up work.} Let us briefly discuss how the ideas
presented in the conference version of this paper \cite{1109590}
influenced later work.  The conference version of this paper posed as
an open question whether for every fixed $k$ there is a
polynomial-time algorithm that, given a hypergraph with fractional
hypertree width $k$, finds a fractional hypertree decomposition having
width $k$ or, at least, having width bounded by a function of
$k$. This question has been partially resolved by the algorithm of
\cite{marx-fractional-talg} that finds a fractional hypertree
decomposition of width $O(k^3)$ if a decomposition of width $k$
exists. This algorithm makes the results of the present paper
stronger, as the polynomial-time algorithms for CSPs with bounded
fractional hypertree width no longer need to assume that a
decomposition is given in the input (see
Section~\ref{sec:algor-appl}). The problem of finding a decomposition
of width exactly $k$, if such a decomposition exists, is still
open. In the special case of hypergraphs whose incidence graphs are
planar, a constant-factor approximation of fractional hypertree width
can be found by exploiting the fact that for such graphs fractional
hypertree width and the tree width of the incidence graph can differ
only by at most a constant factor \cite{DBLP:conf/stacs/FominGT09}.

A core combinatorial idea of the present paper is that Shearer's Lemma
gives an upper bound on the number of solutions of a CSP instance and
we use this bound for the subproblems corresponding to the bags of a
fractional hypertree decomposition. This upper bound has been
subsequently used by \citeN{atserias-focs2008-joins} in
the context of database queries, where it is additionally shown that a
linear programming duality argument implies the tightness of this
bound. We repeat this argument here as Theorem~\ref{th:dual}. An
optimal query evaluation algorithm matching this tight bound was given
by \citeN{Ngo:2012:WOJ:2213556.2213565}. The bound was
generalized to the context of conjunctive queries with functional
dependencies by \citeN{Gottlob:2012:STB:2220357.2220363}.

We show that if $\mathcal H$ is a class of hypergraphs with bounded
fractional hypertree width, then $\CSP(\mathcal H)$ is fixed-parameter
tractable parameterized by the number of variables and, in fact,
polynomial-time solvable (using the algorithm of
\cite{marx-fractional-talg} for finding decompositions). It is a
natural question if there are more general fixed-parameter tractable
or polynomial-time solvable classes of hypergraphs. Very recently,
\citeN{marx-stoc2010-csp} gave a strictly more general such class by
introducing the notion of {\em submodular width} and showing that
$\CSP(\mathcal H)$ is fixed-parameter tractable for classes $\mathcal
H$ with bounded submodular width.  Furthermore, it was shown in
\cite{marx-stoc2010-csp} that there are no classes $\mathcal H$
with unbounded submodular width that make $\CSP(\mathcal H)$
fixed-parameter tractable (the proof uses a complexity-theoretic
assumption called Exponential Time Hypothesis
\cite{MR1894519}). However, with respect to polynomial-time
solvability, bounded fractional hypertree width is still the most
general known tractability condition and it is an open question
whether there are classes $\mathcal H$ with unbounded fractional
hypertree width such that $\CSP(\mathcal H)$ is polynomial-time
solvable.

\section{Preliminaries}
\label{sec:prels}

\subsection{Hypergraphs}
A \emph{hypergraph} is a pair $H = (V(H),E(H))$, consisting of a
set
$V(H)$ of \emph{vertices} and a set $E(H)$ of nonempty subsets of $V(H)$, the
\emph{hyperedges} of $H$. 
We always assume that hypergraphs have no isolated
vertices, that is, for every $v\in V(H)$ there exists at least one $e\in E(H)$
such that $v\in e$.

For a hypergraph $H$ and a set $X\subseteq V(H)$, the \emph{subhypergraph of
  $H$ induced by $X$} is the hypergraph $H[X]=(X,\{e\cap X \mid e \in
E(H)\text{ with }e\cap X\neq\emptyset\})$. We let $H\setminus X=H[V(H)\setminus X]$. The \emph{primal graph} of
a hypergraph $H$ is the graph
\begin{equation*}
\begin{split}
\underline H=(V(H), & \{\{v,w\}\mid v\neq w,\text{ there exists an } \\
& e \in E(H)\text{ such that }\{v,w\}\subseteq e\}).
\end{split}
\end{equation*}
A hypergraph $H$ is \emph{connected} if $\underline H$ is connected. A set
$C\subseteq V(H)$ is \emph{connected (in $H$)} if the induced subhypergraph
$H[C]$ is connected, and a \emph{connected component} of $H$ is a maximal
connected subset of $V(H)$. A sequence of vertices of $H$ is a {\em path} of $H$ if it is
a path of $\underline H$.

A \emph{tree decomposition} of a hypergraph $H$ is a tuple $(T,(B_t)_{t \in
  V(T)})$, where $T$ is a tree and $(B_t)_{t \in V(T)}$ a family of subsets of
$V(H)$ such that for each $e \in E(H)$ there is a node $t \in V(T)$ such that
$e \subseteq B_t$, and for each $v \in V(H)$ the set $\{t\in V(T)\mid v \in
B_t\}$ is connected in $T$. The sets $B_t$ are called the \emph{bags} of the
decomposition. The \emph{width} of a tree-decomposition $(T,(B_t)_{t \in
  V(T)})$ is $\max\big\{|B_t|\bigmid t\in V(t)\}-1$. The \emph{tree width}
$\tw(H)$ of a hypergraph $H$ is the minimum of the widths of all
tree-decompositions of $H$.  It is easy to see that $\tw(H)=\tw(\underline H)$
for all $H$.

It will be convenient for us to view the trees in tree-decompositions as being
rooted and directed from the root to the leaves. For a node $t$ in a (rooted)
tree $T=(V(T),E(T))$, we let $T_t$ be the subtree rooted at $t$, that is, the
induced subtree of $T$ whose vertex set is the set of all vertices reachable
from $t$.

We say that a class $\mathcal H$ of hypergraphs is of \emph{bounded tree
  width} if there is a $k$ such that $\tw(H)\le k$ for all $H\in\mathcal
  H$. We use a similar terminology for other hypergraph invariants.

\subsection{Constraint satisfaction problems}
A \emph{CSP instance} is a triple $I=(V,D,C)$, where $V$ is a set of
\emph{variables}, $D$ is a set called the \emph{domain}, and $C$ is a set of
\emph{constraints} of the form $\langle (v_1,\ldots,v_k), R\rangle$, where
$k\ge 1$ and $R$ is a $k$-ary relation on $D$.
A \emph{solution} to the instance $I$ is an assignment $\alpha:V\to D$ such
that for all constraints $\langle (v_1,\ldots,v_k), R\rangle$ in $C$ we have
$(\alpha(v_1),\ldots,\alpha(v_k))\in R$.

Constraints are specified by explicitly enumerating all possible combinations
of values for the variables, that is, all tuples in the relation $R$.
Consequently, we define the \emph{size} of a constraint $c=\langle
(v_1,\ldots,v_k), R\rangle\in C$ to be the number $\|c\|=k+k\cdot|R|$.  The
\emph{size} of an instance $I=(V,D,C)$ is the number $\|I\|=|V|+|D|+\sum_{c\in
  C}\|c\|$. Of course, there is no need to store a constraint relation
repeatedly if it occurs in several constraints, but this only changes the size
by a polynomial factor. 

Let us make a few remarks about this explicit representation of the
constraints. There are important special cases of constraint
satisfaction problems where the constraints are stored implicitly,
which may make the representation exponentially more succinct.
Examples are Boolean satisfiability, where the constraint relations
are given implicitly by the clauses of a formula in conjunctive normal
form, or systems of arithmetic (in)equalities, where the constraints
are given implicitly by the (in)equalities. However, our
representation is the standard ``generic'' representation of
constraint satisfaction problems in artificial intelligence~(see, for
example, \cite{dec03}). An important application where the constraints
are always given in explicit form is the conjunctive query containment
problem, which plays a crucial role in database query
optimization. Kolaitis and Vardi~\cite{kolvar98} observed that it can
be represented as a constraint satisfaction problem, and the
constraint relations are given explicitly as part of one of the input
queries. A related problem from database systems is the problem of
evaluating conjunctive queries (cf.~Theorem~\ref{theo:cq}). Here the
constraint relations represent the tables of a relational database,
and again they are given in explicit form. The problem of
characterizing the tractable structural restrictions of CSP has also
been studied for other representations of the instances: one can
consider more succinct representations such as disjunctive formulas or
decision diagrams \cite{chegro10} or less succinct representations such
as truth tables \cite{marx-tocs-truthtable}. As the choice of
representation influences the size of the input and the running time
is expressed as a function of the input size, the choice of
representation influences the complexity of the problem and the exact
tractability criterion.

Observe that there is a polynomial-time algorithm deciding whether a given
assignment for an instance is a solution.

The \emph{hypergraph of} the CSP instance $I=(V,D,C)$ is the hypergraph $H_I$
with vertex set $V$ and a hyperedge $\{v_1,\ldots,v_k\}$ for all constraints
$\langle (v_1,\ldots,v_k), R\rangle$ in $C$. For every class $\mathcal H$,
we consider the following decision problem:

\nprob[9]{$\CSP(\mathcal H)$}{A CSP instance $I$ with $H_I\in\mathcal H$}{Decide if $I$
  has a solution}
If the class $\mathcal H$ is not polynomial-time decidable, we view this
as a promise problem, that is, we assume that we are only given instances $I$
with $H_I\in\mathcal H$, and we are only interested in algorithms that work
correctly and efficiently on such instances.

We close this section with a simple example of a class of hypergraphs of
unbounded tree width such that $\CSP(\mathcal H)$ is tractable.

\begin{example}
  Let $\mathcal H$ be that class of all hypergraphs $H$ that have a hyperedge
  that contains all vertices, that is, $V(H)\in E(H)$.  Clearly,
  $\mathcal H$ has unbounded tree width, because the hypergraph $(V,\{V\})$
  has tree width $|V|-1$. We claim that $\CSP(\mathcal H)\in\PTIME$.
  
  To see this, let $I=(V,D,C)$ be an instance of $\CSP(\mathcal H)$. Let
  $\langle (v_1,\ldots,v_k),R\rangle$ be a constraint in $C$ with
  $\{v_1,\ldots,v_k\}=V$. Such a constraint exists because $H_I\in\mathcal
  H$. Each tuple $\bar d=(d_1,\ldots,d_k)\in R$
  completely specifies an assignment $\alpha_{\bar d}$ defined by
  $\alpha_{\bar d}(v_i)=d_i$ for $1\le i\le k$. If for some $i,j$ we have
  $v_i=v_j$, but $d_i\neq d_j$, we leave $\alpha_{\bar d}$ undefined.

  Observe that $I$ is satisfiable if and only if there is a tuple $\bar d\in
  R$ such that $\alpha_{\bar
    d}$ is (well-defined and) a solution for $I$. As $|R|\le\|I\|$, this can
  be checked in polynomial time.
\end{example}

\section{A Polynomial-time algorithm for CSPs with bounded fractional 
cover number}

In this section we prove that if the hypergraph $H_I$ of a CSP instance $I$
has fractional edge cover number $\rho^*(H_I)$, then it can be decided in $\|
I \| ^ {\rho^*(H_I)+O(1)}$ time whether $I$ has a solution. Thus if ${\mathcal
  H}$ is a class of hypergraphs with bounded fractional edge cover number
(that is, there is a constant $r$ such that $\rho^*(H)\le r$ for every $H\in
{\mathcal H}$), then $\CSP({\mathcal H})\in \PTIME$.
Actually, we prove a stronger result: A CSP instance
$I$ has at most $\| I \| ^ {\rho^*(H_I)}$ solutions and all the
solutions can be enumerated in time $\| I \| ^ {\rho^*(H_I)+O(1)}$.

The proof relies on a combinatorial lemma known as Shearer's Lemma. We
use Shearer's Lemma to bound the number of solutions of a CSP
instance; our argument resembles an argument that Friedgut and Kahn
\cite{frikah98} used to bound the number of subhypergraphs of a
certain isomorphism type in a hypergraph. The second author
applied similar ideas in a completely different algorithmic context
\cite{marx-closest-full}.

The \emph{entropy} of
a random variable
$X$ with range $U$ is 
\[
h[X]:=- \sum_{x \in U}
  \Pr(X=x)\log \Pr(X = x)
\]
Shearer's lemma gives an upper bound of a distribution on a product space in
terms of its marginal distributions. 

\begin{lemma}[Shearer's Lemma \cite{chufragrashe86}]\label{lem:shearer}
  Let $X=(X_i\mid i\in I)$ be a random variable, and let $A_j$, for $j\in [m]$, be
  (not necessarily distinct) subsets of the index set $I$ such that each $i\in
  I$ appears in at least $q$ of the sets $A_j$. For every $B\subseteq I$, let
  $X_B=(X_i\mid i\in B)$. Then
  \[
  \sum_{j=1}^m h[X_{A_j}]\ge q\cdot h[X].
  \]
\end{lemma}

Lemma~\ref{lem:shearer} is easy to see in the special case when $q=1$
and $\{A_1,\dots,A_p\}$ is a partition of $V$. The proof of the
general case in \cite{chufragrashe86} is based on the submodularity of
entropy. See also \cite{rha01} for a simple proof.

\begin{lemma}\label{lem:count-bound}
If $I=(V,D,C)$ is a CSP instance where every constraint relation contains at most $N$ tuples, then $I$ has at most $N^
{\rho^*(H_I)} \le \| I \| ^{\rho^*(H_I)}$ solutions.
\end{lemma}

\begin{proof}
Let $x$ be a fractional edge cover of
  $H_I$ with $\sum_{e\in E(H_I)}x(e)=\rho^*(H_I)$; it follows from the
  standard results of linear programming that such an $x$ exists
  with rational values. 
Let $p_e$ and $q$ be 
  nonnegative integers such that $x(e) = p_e/q$. Let $m = \sum_{e\in E(H_I)} p_e$, and let
    $A_1,\ldots,A_m$ be a sequence of subsets of $V$ that contains
    precisely $p_e$ copies of the set $e$, for all $e\in E(H_I)$. Then
    every variable $v \in V$ is contained in at least 
  \[
\sum_{e\in E(H_I):v\in e} p_e
  \;=\; q \cdot \sum_{e\in E(H_I) : v \in e} x(e) \;\geq\; q
  \]
   of the sets $A_i$ (as $x$ is a fractional edge cover).  Let
  $X=(X_v\mid v\in V)$ be uniformly distributed on the solutions of
  $I$, which we assume to be non-empty as otherwise the claim is
  obvious. That is, if we denote by $S$ the number of solutions of
  $I$, then we have $\Pr(X=\alpha)=1/S$ for every solution $\alpha$ of
  $I$.  Then $h[X]=\log S$. We apply Shearer's Lemma to the random
  variable $X$ and the sequence $A_1$, $\dots$, $A_m$ of subsets of $V$.  Assume that $A_i$ corresponds to some
  constraint $\langle (v'_1,\ldots,v'_k), R\rangle$. Then the marginal
  distribution of $X$ on $(v'_1,\ldots,v'_k)$ is $0$ on all tuples not
  in $R$.  Hence the entropy of $X_{A_i}$ is is bounded by the entropy
  of the uniform distribution on the tuples in $R$, that is,
  $h[X_{A_i}]\le \log N$.  Thus by Shearer's Lemma, we have
  \[
  \sum_{e\in E(H_I)} p_e\cdot\log N\ge
\sum_{e\in E(H_I)} p_e \cdot h[X_{e}]=
\sum_{i=1}^m h[X_{A_i}]\ge q\cdot
  h[X]=q\cdot \log S.
  \]
  It follows that 
  \[
  S \le 2^{\sum_{e\in E(H_I)} (p_e/q)\cdot\log N}=2^{\rho^*(H_I)\cdot\log N}=N^{\rho^*(H_I)}.
  \] 
\end{proof}

We would like to turn the upper bound of Lemma~\ref{lem:count-bound}
into an algorithm enumerating all the solutions, but the proof of
Shearer's Lemma is not algorithmic. However, a very simple algorithm
can enumerate the solutions, and Lemma~\ref{lem:count-bound} can be
used to bound the running time of this algorithm. Starting with a
trivial subproblem consisting only of a single variable, the algorithm
enumerates all the solutions for larger and larger subproblems by
adding one variable at a time. To define these subproblems, we need
the following definitions:
\begin{defin}\label{def:projection}
  Let $R$ be an $r$-ary relation over a set $D$. For $1\le i_1<\dots
  <i_\ell \le r$, the {\em projection of $R$} onto the components
  $i_1,\ldots,i_\ell$ is the relation $R^{|i_1,\ldots,i_\ell}$ which
  contains an $\ell$-tuple $(d'_1,\dots,d'_\ell)\in D^\ell$ if and
  only if there is a $k$-tuple $(d_1,\ldots,d_k)\in R$ such that
  $d_j'=d_{i_j}$ for $1\le j\le\ell$.
\end{defin}
Intuitively, a tuple is in $R^{|i_1,\ldots,i_\ell}$ if it can be
extended into a tuple in $R$.
\begin{defin}\label{def:induced-instance}
  Let $I=(V,D,C)$ be a CSP instance and let $V'\subseteq V$ be a
  nonempty subset of variables. The \emph{CSP instance $I[V']$ induced
    by $V'$} is $I'=(V',D,C')$, where $C'$ is defined in the following
  way: for each constraint $c=\langle (v_1,\dots,v_k),R\rangle$ having
  at least one variable in $V'$, there is a corresponding constraint
  $c'$ in $C'$. Suppose that $v_{i_1},\ldots,v_{i_\ell}$ are the
  variables among $v_1,\ldots,v_k$ that are in $V'$. Then the
  constraint $c'$ is defined as $\langle (v_{i_1},\dots,
  v_{i_\ell}),R^{|i_1,\ldots,i_\ell}\rangle$, that is, the relation is
  the projection of $R$ onto the components $i_1,\ldots,i_\ell$.
\end{defin}

Thus an assignment $\alpha$ on $V'$ satisfies $I[V']$ if for each
constraint $c$ of $I$, there is an assignment extending $\alpha$ that
satisfies $c$ (however, it is not necessarily true that there is an
assignment extending $\alpha$ that satisfies every constraint of $I$
simultaneously). Note that that the hypergraph of the induced instance
$I[V']$ is exactly the induced subhypergraph $H_I[V']$.
\begin{theorem}\label{theo:enumerate-solution}
  The solutions of a CSP instance $I$ can be enumerated in time $\| I \| ^
  {\rho^*(H_I)+O(1)}$.
\end{theorem}

\begin{proof}
Let $V=\{v_1,\dots,
  v_n\}$ be an arbitrary ordering of the variables of $I$ and let
  $V_i$ be the subset $\{v_1,\dots, v_i\}$.  
  For $i=1,2,\dots, n$, the algorithm creates a list $L_i$ containing the
  solutions of $I[V_i]$. Since $I[V_n]=I$, the list $L_n$ is exactly what we want.

For $i=1$, the instance $I[V_i]$ has at most $|D|$ solutions, hence the
list $L_i$ is easy to construct.  Notice that a solution of $I[V_{i+1}]$
induces a solution of $I[V_i]$. Therefore, the list $L_{i+1}$ can be
constructed by considering the solutions in $L_i$, extending them to
the variable $v_{i+1}$ in all the $|D|$ possible ways, and checking
whether this assignment is a solution of $I[V_{i+1}]$. Clearly, this can
be done in $|L_i|\cdot|D|\cdot \| I[V_{i+1}] \|^{O(1)}=|L_i|\cdot \| I
\|^{O(1)}$ time. By repeating this procedure for $i=1,2,\dots, n-1$, the list
$L_n$ can be constructed.

The total running time of the algorithm can be bounded by
$\sum_{i=1}^{n-1}|L_i|\cdot \| I \|^{O(1)}$. Observe that
$\rho^*(H_{I[V_i]})\le \rho^*(H_{I})$: $H_{I[V_i]}$ is the
subhypergraph of $H_I$ induced by $V_i$, thus any fractional cover of
the hypergraph of $I$ gives a fractional cover of $I[V_i]$ (for every
edge $e\in E(H_{I[V_i]})$, we set the weight of $e$ to be the sum of
the weight of the edges $e'\in E(H_{I})$ with $e'\cap V_i=e$).
Therefore, by Lemma~\ref{lem:count-bound}, $|L_i|\le \| I \|
^{\rho^*(H_I)}$, and it follows that the total running time is $\| I
\| ^{\rho^*(H_I)+O(1)}$.
\end{proof}

We note that the algorithm of Theorem~\ref{theo:enumerate-solution}
does not actually need a fractional edge cover: the fact that the
hypergraph has small fractional edge cover number is used only in
proving the time bound of the algorithm. By a significantly more
complicated algorithm, Ngo et al.~\cite{Ngo:2012:WOJ:2213556.2213565}
improved Theorem~\ref{theo:enumerate-solution} by removing the $O(1)$
term from the exponent.

\begin{corollary}
  Let $\mathcal H$ be a class of hypergraphs of bounded fractional edge cover
  number. Then $\CSP(\mathcal H)$ is in polynomial time.
\end{corollary}

We conclude this section by pointing out that
Lemma~\ref{lem:count-bound} is tight: there are arbitrarily large
instances $I$ where every constraint relation
contains at most $N$ tuples and the number of solutions is exactly
$N^{\rho^*(H_I)}$.  A similar proof appeared first in
\cite{atserias-focs2008-joins} in the context of database queries, but
we restate it here in the language of CSPs for the convenience of the
reader.

\begin{theorem}\label{th:dual}
  Let $H$ be a hypergraph. For every $N_0\ge 1$, there is a CSP
  instance $I=(V,D,C)$ with hypergraph $H$ where every constraint
  relation contains at most $N\ge N_0$ tuples and $I$ has at least
  $N^{\rho^*(H)}$ solutions.
\end{theorem}
\begin{proof}
  A {\em fractional independent set} of hypergraph $H$ is an assignment
  $y: V(H)\to [0,1]$ such that $\sum_{v\in e}y(v)\le 1$ for every
  $e\in E(H)$. The {\em weight} of $y$ is $\sum_{v\in V(H)}y(H)$. The {\em
    fractional independent set number} $\alpha^*(H)$ is the maximum
  weight of a fractional independent set of $H$. It is a well-known
  consequence of linear-programming duality that
  $\alpha^*(H)=\rho^*(H)$ for every hypergraph $H$, since the two
  values can be expressed by a pair of primal and dual linear programs \cite[Section 30.10]{MR1956925}.

  Let $y$ be a fractional independent set of weight $\alpha^*(H)$. By
  standard results of linear programming, we can assume that $y$ is
  rational, that is, there is an integer $q\ge 1$ such that for every
  $v\in V(H)$, $y(v)=p_v/q$ for some nonnegative integer $p_v$.  
  We define a CSP instance $I=(V,D,C)$ with $V=V(H)$ and $D=[N_0^q]$ such that for every $e\in E(H)$
  where $e=\{v_1,\dots,v_r\}$, there is a constraint
  $\langle(v_1,\dots,v_r),R_e\rangle$ with
\[
R_e=\{ (a_1,\dots,a_r) \mid \text{$a_i\in [N_0^{p_v}]$ for every $1\le i \le r$}\}.
\]
Let $N=N_0^q$. We claim that $R_e$ contains at most $N$ tuples. Indeed, the number of tuples in $R_e$ is exactly
\[
\prod_{v\in{e}}N_0^{p_v}=N_0^{\sum_{v\in{e}}p_v}=
 N_0^{q\cdot {\sum_{v\in{e}}p_v/q}}=(N_0^q)^{\sum_{v\in e}y(v)}\le N_0^q=N,
\]
since $y$ is a fractional independent set.  Observe that $\alpha:
V(H)\to D$ is a solution if and only if $\alpha(v)\in [N_0^{p_v}]$ for
every $v\in V(H)$.  Hence the number of solutions is exactly
  \[
\prod_{v\in{V(H)}}N_0^{p_v}=N_0^{\sum_{v\in{V(H)}}p_v}=
 N_0^{q\cdot {\sum_{v\in{V(H)}}p_v/q}}=(N_0^q)^{\alpha^*(H)}=N^{\alpha^*(H)}=N^{\rho^*(H)},
  \]
as required.
\end{proof}
The significance of this result is that it shows that there is no
``better'' measure than fractional edge cover number that guarantees a
polynomial bound on the number of solutions, in the following formal sense.
Let $w(H)$ be a width measure that guarantees a polynomial bound: that
is, if $I$ is a CSP instance where every relation has at most $N$
tuples, then $I$ has at most $N^{w(H)}$ solutions for some function
$f$. Then by Theorem~\ref{th:dual}, we have $\rho^*(H)\le w(H)$. This
means the upper bound on the number of solutions given by $w(H)$
already follows from the bound given by Lemma~\ref{lem:count-bound}
and hence $\rho^*(H)$ can be considered a stronger measure.

\section{Fractional hypertree decompositions}
Let $H$ be a hypergraph.  A \emph{generalized hypertree decomposition
  of $H$}~\cite{gotleosca02} is a triple $(T,(B_t)_{t \in V(T)},$ $(C_t)_{t \in V(T)})$,
where $(T,(B_t)_{t \in V(T)})$ is a tree decomposition of $H$ and
$(C_t)_{t \in V(T)}$ is a family of subsets of $E(H)$ such that for
every $t\in V(T)$ we have $B_t \subseteq \bigcup C_t$. Here $\bigcup
C_t$ denotes the union of the sets (hyperedges) in $C_t$, that is, the
set $\{v\in V(H)\mid\exists e\in C_t:\;v\in e\}$. We call the sets
$B_t$ the \emph{bags} of the decomposition and the sets $C_t$ the
\emph{guards}. The \emph{width} of $(T,(B_t)_{t \in V(T)},(C_t)_{t \in
  V(T)})$ is $\max\{|C_t|\mid t\in V(T)\}$.  The \emph{generalized
  hypertree width} $\ghw(H)$ of $H$ is the minimum of the widths of
the generalized hypertree decompositions of $H$.  The {\em edge cover
  number} $\rho(H)$ of a hypergraph is the minimum number of edges
needed to cover all vertices; it is easy to see that $\rho(H)\ge
\rho^*(H)$. Observe that the size of $C_t$ has to be at least
$\rho(H[B_t])$ and, conversely, for a given $B_t$ there is always a
suitable guard $C_t$ of size $\rho(H[B_t])$. Therefore, $\ghw(H)\le r$
if there is a tree decomposition where 
$\rho(H[B_t])\le r$ for every $t\in V(T)$.

For the sake of completeness, let us mention that a \emph{hypertree decomposition} of $H$ is a generalized hypertree
decomposition $(T,(B_t)_{t \in V(T)},(C_t)_{t \in V(T)})$ that satisfies the
following additional \emph{special condition}:
$(\bigcup C_t) \cap\bigcup_{u\in V(T_t)}B_{u}\subseteq B_t$ for all $t\in
V(T)$. Recall that $T_t$ denotes the subtree of the $T$ with root $t$. The
\emph{hypertree width} $\hw(H)$ of $H$ is the minimum of the widths of all
hypertree decompositions of $H$. It has been proved in \cite{DBLP:journals/ejc/AdlerGG07} that
$\ghw(H)\le\hw(H)\le3\cdot\ghw(H)+1$. This means that for our purposes,
hypertree width and generalized hypertree width are equivalent. For
simplicity, we will only work with generalized hypertree width.

Observe that for every hypergraph $H$ we have
$\ghw(H)\le\tw(H)+1$. Furthermore, if $H$ is a hypergraph 
with $V(H)\in E(H)$ we have $\ghw(H)=1$
and $\tw(H)=|V(H)|-1$. 

We now give an approximate characterization of (generalized) hypertree
width by a game that is a variant of the \emph{cops and robber} game
\cite{seytho93}, which characterizes tree width: In the \emph{robber
  and marshals game on $H$} \cite{gotleosca03}, a robber plays against
$k$ marshals.  The marshals move on the hyperedges of $H$, trying to
catch the robber. Intuitively, the marshals occupy all vertices of the
hyperedges where they are located. In each move, some of the marshals fly
in helicopters to new hyperedges. The robber moves on the vertices of
$H$. She sees where the marshals will be landing and quickly tries to
escape, running arbitrarily fast along paths of $H$, not being allowed
to run through a vertex that is occupied by a marshal before
\emph{and} after the flight (possibly by two different marshals). The
marshals' objective is to land a marshal via helicopter on a hyperedge
containing the vertex occupied by the robber. The robber tries to
elude capture. The \emph{marshal width} $\mw(H)$ of a hypergraph $H$
is the least number $k$ of marshals that have a winning strategy in
the robber and marshals game played on $H$ (see \cite{adl04} or
\cite{gotleosca03} for a formal definition).

It is easy to see that $\mw(H)\le\ghw(H)$ for every hypergraph $H$. To win the
game on a hypertree of generalized hypertree width $k$, the marshals always
occupy guards of a decomposition and eventually capture the robber at a leaf
of the tree. Conversely, it can be proved that $\ghw(H)\le 3\cdot\mw(H)+1$ \cite{DBLP:journals/ejc/AdlerGG07}.
  
Observe that for every hypergraph $H$, the generalized hypertree width
$\ghw(H)$ is less than or equal to the edge cover number $\rho(H)$:
hypergraph $H$ has a generalized hypertree decomposition consisting of
a single bag containing all vertices and having a guard of size
$\rho(H)$. On the other hand, the following two examples show that
hypertree width and {\em fractional} edge cover number are
incomparable.

\begin{example}\label{exa:ghw<rho*}
  Consider the class of all graphs that only have disjoint edges. The
  tree width and hypertree width of this class is $1$, whereas the
  fractional edge cover number is unbounded.
\end{example}

\begin{example}\label{exa:rho*<ghw}
  For $n\ge 1$, let $H_n$ be the following hypergraph: $H_n$ has a vertex
  $v_S$ for every subset $S$ of $\{1,\ldots,2n\}$ of cardinality
  $n$. Furthermore, for every $i\in\{1,\ldots,2n\}$ the hypergraph $H_n$ has a
  hyperedge $e_i=\{v_S\mid i\in S\}$.
  
  Observe that the fractional edge cover number $\rho^*(H_n)$ is at most $2$,
  because the mapping $x$ that assigns $1/n$ to every hyperedge $e_i$ is a
  fractional edge cover of weight $2$. Actually, it is easy to see that
  $\rho^*(H_n)=2$.
  
  We claim that the hypertree width of $H_n$ is $n$. We show that
  $H_n$ has a hypertree decomposition of width $n$.  Let
  $S_1=\{1,\dots,n\}$ and $S_2=\{n+1,\dots,2n\}$. We construct a
  generalized hypertree decomposition for $H_n$ with a tree $T$ having
  two nodes $t_1$ and $t_2$. For $i=1,2$, we let $B_{t_1}$ contain a
  vertex $V_S$ if and only if $S\cap S_i\neq \emptyset$. For each edge
  $e_j\in E(H_n)$, there is a bag of the decomposition that contains
  $e_j$: if $j\in S_i$, then $B_{t_i}$ contains every vertex of
  $e_j$. We set the guard $C_{t_i}$ to contain every $e_j$ with $j\in
  S_i$. It is clear that $|C_{t_i}|=n$ and $C_{t_i}$ covers $B_{t_i}$:
  vertex $v_S$ is in $B_{t_i}$ only if there is a $j\in S\cap S_i$, in
  which case $e_j\in C_{t_i}$ covers $v_S$.  Thus this is indeed a
  generalized hypertree decomposition of width $n$ for $H_n$ and
  $\ghw(H_n)\le n$ follows.

To see that $\ghw(H_n)>n-1$, we argue that the robber
  has a winning strategy against $(n-1)$ marshals in the robber and marshals
  game.  Consider a position of the game where the marshals occupy edges
  $e_{j_1},\ldots,e_{j_{n-1}}$ and the robber occupies a vertex $v_S$ for a
  set $S$ with $S\cap\{j_1,\ldots,j_{n-1}\}=\emptyset$. Suppose that in the
  next round of the game the marshals move to the edges
  $e_{k_1},\ldots,e_{k_{n-1}}$. Let $i\in
  S\setminus\{k_1,\ldots,k_{n-1}\}$. The robber moves along the edge $e_i$ to
  a vertex $v_{R}$ for a set
  $R\subseteq\{1,\ldots,2n\}\setminus\{k_1,\ldots,k_{n-1}\}$ of cardinality
  $n$ that contains $i$. If she plays this way, she can never be captured.
\end{example}

For a hypergraph $H$ and a mapping $\gamma:E(H)\to[0,\infty)$, we let
\[
B(\gamma)=\{v\in
V(H)\mid\sum_{e\in E(H),v\in e}\gamma(e)\ge 1\}.
\] 
We may think of $B(\gamma)$
as the set of all vertices ``blocked'' by $\gamma$.
Furthermore, we let 
$\weight(\gamma)=\sum_{e\in E}\gamma(e)$.

\begin{defin}
  Let $H$ be a hypergraph. A \emph{fractional hypertree decomposition of $H$}
  is a triple $(T,$ $(B_t)_{t \in V(T)}$, $(\gamma_t)_{t \in V(T)})$, where $(T,(B_t)_{t \in V(T)})$ is a tree decomposition of $H$ and
  $(\gamma_t)_{t \in V(T)}$ is a family of mappings from $E(H)$ to
  $[0,\infty)$ such that for every $t\in V(T)$ we have $B_t\subseteq
  B(\gamma_t)$. 
  
  We call the sets $B_t$ the \emph{bags} of the decomposition and the mappings
  $\gamma_t$ the \emph{(fractional) guards}.
  
  The \emph{width} of $(T,(B_t)_{t \in V(T)},(\gamma_t)_{t \in V(T)})$ is
  $\max\{\weight(\gamma_t)\mid t\in V(T)\}$. The \emph{fractional hypertree
    width} $\fhw(H)$ of $H$ is the minimum of the widths of the fractional
  hypertree decompositions of $H$. Equivalently, $\fhw(H)\le r$ if $H$ has a tree decomposition where $\rho^*(B_t)\le r$ for every bag $B_t$.
\end{defin}

It is easy to see  that the minimum of the widths of all fractional hypertree
decompositions of a hypergraph $H$ always exists and is rational. This follows
from the fact that, up to an obvious equivalence, there are only finitely many
tree decompositions of a hypergraph.

Clearly, for every hypergraph $H$ we have 
\[
\fhw(H)\le \rho^*(H)\quad\text{and}\quad\fhw(H)\le\ghw(H).
\]
Examples \ref{exa:ghw<rho*} and \ref{exa:rho*<ghw} above show that there are
families of hypergraphs of bounded fractional hypertree width, but unbounded
fractional edge cover number and unbounded generalized hypertree width.

It is also worth pointing out that for every hypergraph $H$,
\[
\fhw(H)=1\iff\ghw(H)=1.
\]
To see this, note that if $\gamma:E(H)\to[0,\infty)$ is a mapping with
$\weight(\gamma)=1$ and $B\subseteq B(\gamma)$, then $B\subseteq e$ for all
$e\in E(H)$ with $\gamma(e)>0$. Thus instead of using $\gamma$ as a guard in a
fractional hypertree decomposition,
we may use the integral guard $\{e\}$ for any $e\in E(H)$ with
$\gamma(e)>0$.
Let us remark that $\ghw(H)=1$ if and only if $H$ is acyclic
\cite{gotleosca02}.

\subsection{The robber and army game}
As robbers are getting ever more clever, it takes more and more powerful
security forces to capture them.  In the \emph{robber and army game on a
  hypergraph $H$}, a robber plays against a general commanding an army of $r$
battalions of soldiers. The general may distribute his soldiers arbitrarily on
the hyperedges. However, a vertex of the hypergraph is only blocked if the
number of soldiers on all hyperedges that contain this vertex adds up to the
strength of at least one battalion. The game is then played like the robber
and marshals game.

\begin{defin}
  Let $H$ be a hypergraph and $r$ a nonnegative real. The \emph{robber
    and army game on $H$ with $r$ battalions} (denoted by $\RA(H,r)$)
  is played by two players, the \emph{robber} and the
  \emph{general}. A \emph{position} of the game is a pair
  $(\gamma,v)$, where $v\in V(H)$ and $\gamma:E(H)\to[0,\infty)$ with
  $\weight(\gamma)\le r$. To start a game, the robber
  picks an arbitrary $v_0$, and the initial position is $(0,v_0)$,
  where $0$ denote the constant zero mapping.

  In each round, the players move from the current position
  $(\gamma,v)$ to a new position $(\gamma',v')$ as follows: The general
  selects $\gamma'$, and then the robber selects $v'$ such that there is a
  path from $v$ to $v'$ in the hypergraph 
  $
  H\setminus(B(\gamma)\cap B(\gamma')\big).
  $

  If a position $(\gamma,v)$ with $v\in B(\gamma)$ is reached, the play ends
  and the general wins.  If the play continues forever, the robber
  wins.
  
  The \emph{army width} $\aw(H)$ of $H$ is the least $r$ such that the general
  has winning strategy for the game $\RA(H,r)$.
\end{defin}

Again, it is easy to see that $\aw(H)$ is well-defined and rational
(observe that two positions $(\gamma_1,v)$ and $(\gamma_2,v)$ are
equivalent if $B(\gamma_1)=B(\gamma_2)$ holds).

\begin{theorem}\label{theo:game}
  For every hypergraph $H$,
  \[
  \aw(H)\le\fhw(H)\le3\cdot\aw(H)+2.
  \]
\end{theorem}

The rest of this subsection is devoted to a proof of this theorem.
The proof is similar to the proof of the corresponding result for the robber
and marshal game and generalized hypertree width in \cite{DBLP:journals/ejc/AdlerGG07}, which
in turn is based on ideas from \cite{ree97,seytho93}.

Let $H$ be a hypergraph and $\gamma,\sigma:E(H)\to[0,\infty)$. For a set
$W\subseteq V(H)$, we let
\[
\weight(\gamma|W)=\sum_{\substack{e\in E(H)\\ e\cap
    W\neq\emptyset}}\gamma(e).
\]
A mapping $\sigma:E(H)\to[0,\infty)$ is a
\emph{balanced separator for $\gamma$} if for every connected component $R$ of
$H\setminus B(\sigma)$,
\[
\weight(\gamma|R)\le\frac{\weight(\gamma)}{2}.
\]

\begin{lemma}\label{lem:sep}
  Let $H$ be a hypergraph with $\aw(H)\le r$ for some nonnegative real $r$. Then every $\gamma:E(H)\to[0,\infty)$ has a balanced separator
  of weight $r$.
\end{lemma}

\begin{proof}
Suppose for contradiction that $\gamma:E(H)\to[0,\infty)$ has no balanced separator
of weight $r$. We claim that the robber has a winning strategy for the game
$\RA(H,r)$. The robber simply maintains the invariant that in every position
$(\sigma,v)$ of the game, $v$ is contained in the connected component $R$ of
$H\setminus B(\sigma)$ with $\weight(\gamma|R)>\weight(\gamma)/2$.  

To see that this is possible, let $(\sigma,v)$ be such a position. Suppose
that the general moves from $\sigma$ to $\sigma'$, and let $R'$ be the
connected component of $H\setminus B(\sigma')$ with
$\weight(\gamma|R')>\weight(\gamma)/2$. Then there must be some $e\in E(H)$
such that $e\cap R\neq\emptyset$ and $e\cap R'\neq\emptyset$, because
otherwise we had
\[
\weight(\gamma)=\weight(\gamma)/2+\weight(\gamma)/2<\weight(\gamma|R)+\weight(\gamma|R')\le\weight(\gamma),
\]
which is impossible. Thus the robber can move from $R$ to $R'$ via the edge
$e$.
\end{proof}

Let $H$ be a hypergraph and $H'$ an induced subhypergraph of $H$. Then the
\emph{restriction} of a mapping $\gamma:E(H)\to[0,\infty)$ to $H'$ is the
mapping $\gamma':E(H')\to [0,\infty)$ defined by
\[
\gamma'(e')=\sum_{\substack{e\in E(H)\\e\cap V(H')=e'}}\gamma(e).
\]
Note that $\weight(\gamma')\le\weight(\gamma)$ and $B(\gamma')=B(\gamma)\cap V(H')$. The inequality may be strict
because edges with nonempty weight may have an empty intersection with $V(H')$.
Conversely, the \emph{canonical extension} of a mapping $\gamma':E(H')\to[0,\infty)$ to $H$ is the
mapping $\gamma:E(H)\to [0,\infty)$ defined by
\[
\gamma(e)=\frac{\gamma'(e\cap V(H'))}{|\{e_1\in E(H)\mid e_1\cap V(H')=e\cap V(H)\}|}
\]
if $e\cap V(H')\neq\emptyset$ and $\gamma(e)=0$ otherwise.
Intuitively, for every $e'\in E(H')$, we distribute the weight of $e'$
equally among all the edges $e\in E(H)$ whose intersection with
$V(H')$ is exactly $e'$.  Note that $\weight(\gamma)=\weight(\gamma')$
and $B(\gamma')=B(\gamma)\cap V(H')$.

\begin{proof}[(of Theorem~\ref{theo:game})] Let $H$ be a
hypergraph.  To prove that $\aw(H)\le\fhw(H)$, let $(T,(B_t)_{t\in
  V(T)},(\gamma_t)_{t\in V(T)})$ be a fractional hypertree decomposition of
$H$ having width $\fhw(H)$. We claim that the general has a winning strategy for $\RA(H,r)$. Let
$(0,v_0)$ be the initial position. The general plays in such a way that all
subsequent positions are of the form $(\gamma_t,v)$ such that $v\in B_u$ for
some $u\in V(T_t)$. Intuitively, this means that the robber is trapped in the
subtree below $t$. Furthermore, in each move the general reduces the height of
$t$. He starts by selecting $\gamma_{t_0}$ for the root ${t_0}$ of $T$.
Suppose the game is in a position $(\gamma_t,v)$ such that $v\in B_u$ for some
$u\in V(T_t)$. If $u=t$, then the robber has lost the game. So let us assume
that $u\neq t$. Then there is a child $t'$ of $t$ such that $u\in V(T_{t'})$.
The general moves to $\gamma_{t'}$. Suppose the robber escapes to a $v'$ that
is not contained in $B_{u'}$ for any $u'\in T_{t'}$. Then there is a path from
$v$ to $v'$ in $H\setminus(B(\gamma_t)\cap B(\gamma_{t'}))$ and hence in
$H\setminus(B_t\cap B_{t'})$. However, it follows easily from the fact that
$(T,(B_t)_{t\in T})$ is a tree decomposition of $H$ that every path from a bag
in $T_{t'}$ to a bag in $T\setminus T_{t'}$ must intersect $B_t\cap B_{t'}$.  This
proves that $\aw(H)\le\fhw(H)$.

\medskip
For the second inequality, we shall prove the following stronger claim:

\medskip\noindent \textit{Claim: }Let $H$ be a hypergraph with
$\aw(H)\le r$ for some nonnegative real $r$. Furthermore, let
$\gamma:E(H)\to[0,\infty)$ such that $\weight(\gamma)\le 2r+2$. Then
there exists a fractional hypertree decomposition of $H$ of width at
most $3r+2$ such that $B(\gamma)$ is contained in the bag of the root
of this decomposition.

\medskip
Note that for $\gamma=0$, the claim yields the desired fractional
hypertree decomposition of $H$.

\medskip\noindent \textit{Proof of the claim: }The proof is by induction on
the cardinality of $V(H)\setminus B(\gamma)$.

By Lemma~\ref{lem:sep}, there is a
balanced separator of weight at most $r$ for $\gamma$ in $H$.
Let $\sigma$ be such a separator, and define $\chi:E(H)\to[0,\infty)$ by
$\chi(e)=\gamma(e)+\sigma(e)$. Then $\weight(\chi)\le 3r+2$, and
$B(\gamma)\cup B(\sigma)\subseteq B(\chi)$.

If $V(H)=B(\chi)$ (this is the induction basis), then the 1-node
decomposition with bag $V(H)$ and guard $\chi$ is a fractional hypertree
decomposition of $H$ of width at most $3r+2$.

Otherwise, let $R_1,\ldots,R_m$ be the connected components of
$H\setminus B(\chi)$. Note that we cannot exclude the case $m=1$ and
$R_1=V(H)\setminus B(\chi)$.

For $1\le i\le m$, let $e_i$ be an edge of $H$ such that $e_i\cap
R_i\neq\emptyset$, and let $S_i$ be the unique connected component of $H\setminus
B(\sigma)$ with $R_i\subseteq S_i$. Note that $\weight(\gamma|S_i)\le r+1$, because $\sigma$ is a balanced separator for
$\gamma$.
Let $\chi_i:E(H)\to[0,\infty)$ be defined by
\[
\chi_i(e)=
\begin{cases}
  1&\text{if }e=e_i,\\
  \sigma(e)+\gamma(e)&\text{if }e\neq e_i\text{ and }S_i\cap e\neq\emptyset,\\
  \sigma(e)&\text{otherwise}.
\end{cases}
\]
Then
\[
\weight(\chi_i)\le 1+\weight(\sigma)+\weight(\gamma|S_i)\le 2r+2
\]
and $B(\chi_i)\setminus R_i\subseteq B(\chi)$ (as $e_i$ cannot
intersect any $R_j$ with $i\neq j$).  Let $H_i=H[R_i\cup B(\chi_i)]$
and observe that
\[
V(H_i)\setminus
B(\chi_i)\subseteq R_i\setminus e_i\subset R_i\subseteq V(H)\setminus 
B(\gamma)
\]
(the first inclusion holds because $\chi_i(e_i)=1$).
Thus the induction hypothesis is applicable to $H_i$ and the restriction of
$\chi_i$ to $H_i$.
It yields a fractional hypertree decomposition $(T^i,(B^i_t)_{t\in
  V(T^i)},(\gamma^i_t)_{t\in V(T^i)})$ of $H_i$ of weight at most $3r+2$ such
that $B(\chi_i)$ is contained in the bag $B^i_{t_0^i}$ of the
root
$t_0^i$ of $T^i$.

Let $T$ be the disjoint union of $T^1,\ldots,T^m$ together with a new root
$t_0$ that has edges to the roots $t_0^i$ of the $T^i$. Let $B_{t_0}= B(\chi)$ and
$B_t=B_t^i$ for all $t\in V(T^i)$. Moreover, let $\gamma_{t_0}= \chi$, and let
$\gamma_t$ be the canonical extension of $\gamma_t^i$ to $H$ for all $t\in
V(T^i)$.

It remains to prove that $(T,(B_t)_{t\in V(T)},(\gamma_t)_{t\in V(T)})$ is a
fractional hypertree decomposition of $H$ of width at most $3r+2$.
Let us first verify that $(T,(B_t)_{t\in V(T)})$ is a tree decomposition.
\begin{itemize}
\item Let $v\in V(H)$. To see that $\{v\in V(T)\mid v\in B_t\}$ is
  connected in $T$, observe that $\{t\in V(T^i)\mid v\in B_{t_i}\}$ is
  connected (maybe empty) for all $i$.  If $v\in R_i$ for some $i$,
  then $v\not\in V(H_j)=R_j\cup B(\chi_j)$ for any $i\neq j$ (as $R_i$
  and $R_j$ are disjoint and we have seen that $B(\chi_j)\setminus
  R_j\subseteq B(\chi)$) and this this already shows that $\{t\in
  V(T)\mid v\in B_t\}$ is connected. Otherwise, $v\in
  B(\chi_i)\setminus R_i\subseteq B(\chi)=B_{t_0}$ for all $i$ such
  that $v\in V(H_i)$. Again this shows that $\{v\in V(T)\mid v\in
  B_t\}$ is connected.
\item Let $e\in E(H)$. Either $e\subseteq B(\chi)=B_{t_0}$, or there
  is exactly one $i$ such that $e\subseteq R_i\cup B(\chi_i)$. In the
  latter case, $e\subseteq B_t$ for some $t\in V(T^i)$.
\end{itemize}
It remains to prove that $B_t\subseteq B(\gamma_t)$ for all $t\in T$. For the
root, we have $B_{t_0}= B(\gamma_{t_0})$. For $t\in V(T^i)$, we have $B_t\subseteq
B(\gamma_t^i)= B(\gamma_t)\cap V(H_i)\subseteq B(\gamma_t)$. Finally, note
that $\weight(\gamma_t)\le 3r+2$ for all $t\in V(T)$. This completes the proof
of the claim.  
\end{proof}

\begin{remark}
  With respect to the difference between hypertree decompositions and
  generalized hypertree decompositions, it is worth observing that the
  fractional tree decomposition $(T,(B_t)_{t\in V(T)},(\gamma_t)_{t\in V(T)})$
  of width at most $3r+2$ constructed in the proof of the theorem satisfies
  the following special condition: $B(\gamma_t) \cap\bigcup_{u\in
    V(T_t)}B_{u}\subseteq B_t$ for all $t\in V(T)$.
  This implies that a hypergraph of fractional hypertree width at most $r$ has
  a fractional hypertree decomposition of width at most $3r+2$ that satisfies
  the special condition. 
\end{remark}

\subsection{Finding decompositions}
For the algorithmic applications, it is essential to have algorithms
that find fractional hypertree decompositions of small width. The
question is whether for any fixed $r>1$ there is a polynomial-time
algorithm that, given a hypergraph $H$ with $\fhw(H)\le r$, computes a
fractional hypertree decomposition of $H$ of width at most $r$ or maybe of
width at most $f(r)$ for some function $f$.  Similarly to hypertree width, one
way of obtaining such an algorithm would be through the army and
robber game characterization. The idea would be to inductively compute
the set of all positions of the game from which the general wins in
$0,1,\ldots$ rounds. The problem is that, as opposed to the robber and
marshals game, there is no polynomial bound on the number of
positions.

Using a different approach (approximately solving the problem of
finding balanced separators with small fractional edge cover number),
\citeN{marx-fractional-talg} gave an algorithm that approximates fractional hypertree width in the following sense:
\begin{theorem}[\cite{marx-fractional-talg}]\label{th:approx}
For every $r\ge 1$, there is an $n^{O(r^3)}$ time algorithm that, given a hypergraph with fractional hypertree width at most $r$, finds a fractional hypertree decomposition of width $O(r^3)$.
\end{theorem}
The main technical challenge in the proof of Theorem~\ref{th:approx}
is finding a separator with bounded fractional edge cover number
that separates two sets $X$, $Y$ of vertices. An approximation
algorithm is given in \cite{marx-fractional-talg} for this problem, which
finds a separator of weight $O(r^3)$ if a separator of weight
$r$ exists. This algorithm is used to find balanced separators, which
in turn is used to construct a tree decomposition (by an argument
similar to the proof of the second part of Theorem~\ref{theo:game}).

It is shown in \cite{marx-fractional-talg,DBLP:conf/stacs/FominGT09} that deciding whether $H$
has fractional hypertree width $r$ is NP-hard if $r$ is part of the
input. However, the more relevant question of whether for every fixed
$r$, a fractional hypertree decomposition of width $r$ can be found in
polynomial-time (i.e., if the width bound $O(r^3)$ in Theorem~\ref{th:approx} can be
improved to $r$) is still open. Given that it is NP-hard to decide
whether a hypergraph has generalized hypertree width at most 3
\cite{DBLP:journals/jacm/GottlobMS09}, it is natural to expect that a
similar hardness result holds for fractional hypertree width as well.

\subsection{Algorithmic applications}\label{sec:algor-appl}
In this section, we discuss how problems can be solved by fractional
hypertree decompositions of bounded width.  First we give a basic
result, which formulates why tree decompositions and width measures
are useful in the algorithmic context: CSP can be efficiently solved
if we can polynomially bound the number of solutions in the bags.
Recall that $H_I$ denotes the hypergraph of a CSP instance $I$. If
$(T,(B_t)_{t \in V(T)})$ is a tree decomposition of $H_I$, then
$I[B_t]$ denotes the instance induced by bag $B_t$, see
Definition~\ref{def:induced-instance}.
\begin{lemma}\label{lem:alg-bags}
  There is an algorithm that, given a CSP instance, a hypertree
  decomposition $(T,(B_t)_{t \in
  V(T)})$ of $H_I$, and for every $t\in V(t)$ a list $L_t$ of
  all solutions of $I[B_t]$, decides in time $C\cdot \| I \| ^{O(1)}$ if $I$ is satisfiable (and
  computes a solution if it is), where $C:=\max_{t\in V(T)}|L_t|$.
\end{lemma}
\begin{proof} Define $V_t:=\bigcup_{t\in V(T_t)}B_t$. For each $t\in V(T)$,
our algorithm constructs the list $L'_t\subseteq L_t$ of those
solutions of $I[B_t]$ that can be extended to a solution of
$I[V_t]$. Clearly, $I$ has a solution if and only if $L'_{t_0}$ is not
empty for the root $t_0$ of the tree decomposition.

The algorithm proceeds in a bottom-up manner: when constructing the
list $L'_t$, we assume that for every child $t'$ of $t$, the lists
$L'_{t'}$ are already available.  If $t$ is a leaf node, then
$V_t=B_t$, and $L'_t=L_t$.  Assume now that $t$ has children $t_1$,
$\dots$, $t_k$. We claim that a solution $\alpha$ of $I[B_t]$ can be
extended to $I[V_t]$ if and only if for each $1\le i \le k$, there is
a solution $\alpha_i$ of $I[V_{t_i}]$ that is compatible with $\alpha$
(that is, $\alpha$ and $\alpha_i$ assign the same values to the
variables in $B_t\cap V_{t_i}=B_t\cap B_{t_i}$). The necessity of this
condition is clear: the restriction of a solution of $I[V_t]$ to
$V_{t_i}$ is clearly a solution of $I[V_{t_i}]$. For sufficiency,
suppose that the solutions $\alpha_i$ exist for every child $t_i$. They can be combined to
an assignment $\alpha'$ on $V_t$ extending $\alpha$ in a well-defined
way: every variable $v\in V_{t_i}\cap V_{t_j}$ is in $B_t$, thus
$\alpha_i(v)=\alpha_j(v)=\alpha(v)$ follows for such a variable. Now $\alpha'$ is a
solution of $I[V_{t}]$: for each constraint of $I[V_t]$, the variables
of the constraint are contained either in $B_t$ or in $V_{t_i}$ for
some $1\le i \le k$, thus $\alpha$ or $\alpha_i$ satisfies the
constraint, implying that $\alpha'$ satisfies it as well.

Therefore, $L'_t$ can be determined by first enumerating every
solution $\alpha\in L_t$, and then for each $i$, checking whether
$L'_{t_i}$ contains an assignment $\alpha_i$ compatible with
$\alpha$. This check can be efficiently performed the following
way. Recall that solutions $\alpha$ and $\alpha_i$ are compatible if
their restriction to $B_t\cap B_{t_i}$ is the same
assignment. Therefore, after computing $L'_{t_i}$, we restrict every
$\alpha_i\in L'_{t_i}$ to $B_t\cap B_{t_i}$ and store these
restrictions in a trie data structure for easy membership tests. Then
to check if there is an $\alpha_i\in L'_{t_i}$ compatible with
$\alpha$, all we need to do is to check if the restriction of $\alpha$
to $B_t\cap B_{t_i}$ is in the trie corresponding to $L'_{t_i}$, which
can be checked in $\| I \|^{O(1)}$.  Thus $L'_t$ (and the
corresponding trie structure) can be computed in time $|L_t|\cdot \| I
\|^{O(1)}$.  As every other part of the algorithm can be done in time
$\| I \|^{O(1)}$, it follows that the total running time can be
bounded by $C\cdot \|I\|^{O(1)}$.  Using standard bookkeeping
techniques, it is not difficult to extend the algorithm such that it
actually returns a solution if one exists.
\end{proof}

Lemma~\ref{lem:alg-bags} tells us that if we have a tree decomposition
where we can give a polynomial bound on the number of solutions in the
bags for some reason (and we can enumerate all these solutions), then
the problem can be solved in polynomial time.  Observe that in a
fractional hypertree decomposition every bag has bounded fractional
edge cover number and hence Theorem~\ref{theo:enumerate-solution} can
be used to enumerate all the solutions. It follows that if a
fractional hypertree decomposition of bounded width is given in the
input, then the problem can be solved in polynomial time. Moreover, if
we know that the hypergraph has fractional hypertree width at most $r$
(but no decomposition is given in the input), then we can use brute
force to find a fractional hypertree decomposition of width at most
$r$ by trying every possible decomposition and then solve the problem
in polynomial time.  This way, the running time is polynomial in the
input size times an (exponential) function of the number of variables.
This immediately shows that if we restrict CSP to a class of hypergraphs whose
fractional hypertree width is at most a constant $r$, then the problem
is fixed-parameter tractable parameterized by the number of
variables. To get rid of the exponential factor depending on the
number of variables and obtain a polynomial-time algorithm, we can
replace the brute force search for the decomposition by the
approximation algorithm of
Theorem~\ref{th:approx}~\cite{marx-fractional-talg}.

\begin{theorem}\label{theo:alg-cspfhw}
  Let $r\ge 1$. Then there is a polynomial-time algorithm that, given a CSP
  instance $I$ of fractional hypertree width
  at most $r$, decides if $I$ is satisfiable (and computes a solution if it
  is).
\end{theorem}

\begin{proof} Let $I$ be a CSP instance of fractional hypertree width at most
$r$, and let $(T, (B_t)_{t \in V(T)}, (\gamma_t)_{t \in V(T)})$ be the
fractional hypertree decomposition of $\mathcal H_I$ of width $O(r^3)$
computed by the algorithm of Theorem~\ref{th:approx}. By the
definition, the hypergraph of $I[B_t]$ has fractional edge cover
number $O(r^3)$ for every bag $B_t$. Thus by
Theorem~\ref{theo:enumerate-solution}, the list $L_t$ of the solutions
of $I[B_t]$ has size at most $\| I \|^{O(r^3)}$ and can be determined
in time $\| I\| ^{O(r^3)}$. Therefore, we can find a solution in time
$\|I \|^{O(r^3)}$ using the algorithm of Lemma~\ref{lem:alg-bags}.
\end{proof}

In the remainder of this section, we sketch further algorithmic
applications of fractional hypertree decompositions. As these results
follow from our main results with fairly standard techniques, we omit
a detailed and technical discussion.

It is has been observed by \citeN{fedvar98} that
constraint satisfaction problems can be described as homomorphism
problems for relational structures (see \cite{fedvar98} or
\cite{DBLP:journals/jacm/Grohe07} for definitions and
details). A \emph{homomorphism} from a structure $A$ to a
structure $B$ is a mapping from the domain of $A$ to the domain of $B$
that preserves membership in all relations. With each structure $A$ we
can associate a hypergraph $H_A$ whose vertices are the elements of
the domain of $A$ and whose hyperedges are all sets
$\{a_1,\ldots,a_k\}$ such that $(a_1,\ldots,a_k)$ is a tuple in some
relation of $A$. For every class $\mathcal H$ of hypergraphs, we let
$\HOM(\mathcal H)$ be the problem of deciding whether a structure $A$
with $H_A\in\mathcal H$ has a homomorphism to a structure $B$. As an
immediate corollary to Theorem~\ref{theo:alg-cspfhw}, we obtain:

\begin{corollary}\label{cor:hom}
  Let $\mathcal H$ be a class of hypergraphs of bounded fractional
  hypertree width. Then $\HOM(\mathcal H)$ is solvable in polynomial time.
\end{corollary}

An \emph{embedding} is a homomorphism that is one-to-one. Note that there is
an embedding from a structure $A$ to a structure $B$ if and only if $B$ has a
substructure isomorphic to $A$. Analogously to $\HOM(\mathcal H)$, we
define the problem $\EMB(\mathcal H)$ of deciding
whether a structure $A$ with $H_A\in\mathcal H$ has an embedding into
a structure $B$. Observe that $\EMB(\mathcal H)$ is NP-complete even
if $\mathcal H$ is the class of paths, which has fractional hypertree
width $1$, because the \textsc{Hamiltonian Path} problem is a special
case. However, we obtain a fixed-parameter tractability result for $\EMB(\mathcal H)$ parameterized by the size
$||A||$ of the input structure $A$:

\begin{theorem}
  Let $\mathcal H$ be a class of hypergraphs of bounded fractional hypertree
  width. Then $\EMB(\mathcal H)$ parameterized by the size of the input
  structure $A$ is fixed-parameter tractable. More precisely, there is an
  algorithm that, given a structure $A$ with $H_A\in\mathcal H$ and a
  structure $B$, decides if there is an embedding of $A$ into $B$ in
  time $2^{||A||^{O(1)}}||B||^{O(1)}$.
\end{theorem}
This follows from Corollary~\ref{cor:hom} with Alon, Yuster, and Zwick's
\cite{aloyuszwi95} color coding technique.

In some situations it is necessary to not only decide whether a
CSP-instance has a solution or to compute one solution, but to
enumerate all solutions. As the number of solutions may be exponential
in the instance size, we can rarely expect a polynomial-time algorithm
for this problem. Instead, we may ask for a \emph{polynomial-delay
  algorithm}, which is required to compute the first solution in
polynomial time and after returning a solution is required to return
the next solution (or determine that no other solution exists) in
polynomial time. Polynomial-delay algorithms for CSPs have been
studied more systematically in \cite{bdgm11+}.

\begin{theorem}
    Let $r\ge 1$. Then there is a polynomial-delay algorithm that, given a CSP
  instance $I$ of fractional hypertree width
  at most $r$, enumerates all solutions of $I$.
\end{theorem}

\begin{proof}[(sketch)]\sloppy
  Given an instance $I$, the algorithm first computes a fractional
  hypertree decomposition $(T, (B_t)_{t \in V(T)}, (\gamma_t)_{t \in
    V(T)})$ of $\mathcal H_I$ of width $O(r^3)$ using the algorithm of
  Theorem~\ref{th:approx}. Then it orders the nodes of $T$ by a
  preorder traversal and then orders the variables of $I$ in such a
  way that if $v\in B_t$ and $w\in B_u\setminus B_t$ and $t$ comes
  before $u$ in the preorder traversal, $v$ is smaller than $w$. For
  every node $t$ of $T$ the algorithm computes a list of all solutions
  for $I[B_t]$ and sorts it lexicographically. It is easy to modify
  the algorithm of Lemma~\ref{lem:alg-bags} to compute the
  lexicographically first solution and, for every given solution, the
  lexicographically next solution.
\end{proof}

A problem closely related to the problem of enumerating all solutions
to a given CSP-instance is the problem of computing the answer for a
conjunctive query in a relational database (see~\cite{kolvar98}). To
be precise, answering conjunctive queries is equivalent to computing
projections of solution sets of CSP-instances to a given subset $V'$ of variables. With each conjunctive
query, we can associate hypergraph in a similar way as we did for
CSP-instances. Then answering queries with hypergraphs in $\mathcal H$
is equivalent to computing projections of solution sets of
$\CSP(\mathcal H)$-instances. We define the \emph{fractional hypertree width} of
a conjunctive query to be the fractional hypertree width of its
hypergraph. 

\begin{theorem}\label{theo:cq}
    Let $r\ge 1$. Then there is a polynomial-delay algorithm that,
    given a conjunctive query $Q$ of fractional hypertree width at
    most $r$ and a database instance $D$, computes the answer $Q(D)$
    of $Q$ in $D$.
\end{theorem}

The proof is a straightforward refinement of the proof of the
previous theorem.

\section{Conclusions}
In this paper we have considered structural properties that can make a
constraint satisfaction problem polynomial-time solvable. Previously,
bounded hypertree width was the most general such property. Answering
an open question raised in \cite{chedal05,DBLP:journals/jcss/CohenJG08,ggmss05,DBLP:journals/jacm/Grohe07},
we have identified a new class of polynomial-time solvable CSP
instances: instances having bounded fractional edge cover number. This
result suggests the definition of fractional hypertree width, which is
always at most as large as the hypertree width (and in some cases much
smaller). It turns out that CSP is polynomial-time solvable for
instances having bounded fractional hypertree width, if the hypertree
decomposition is given together with the instance. This immediately
implies that CSP is fixed-parameter tractable parameterized by the
number of variables for hypergraphs with bounded fractional
hypertree width. Furthermore, together with the algorithm of
\cite{marx-fractional-talg} finding approximate fractional hypertree
decompositions, it also follows that CSP is polynomial-time solvable
for this class. Currently, bounded fractional hypertree width is the most general known structural property 
that makes CSP polynomial-time solvable.

The most natural open question we leave open regarding fractional
hypertree width is whether for every fixed $r$ there is a
polynomial-time algorithm that decides if a hypergraph has fractional
hypertree width at most $r$ (and if so, constructs a
decomposition). As the analogous problem for generalized hypertree
width is NP-hard for $r=3$, we expect this to be the case for
fractional hypertree width as well.

Another open question is whether there are polynomial-time solvable or
fixed-parameter tractable families of CSP instances having unbounded
fractional hypertree width. Very recently, \cite{marx-stoc2010-csp}
showed that, under suitable complexity assumptions, bounded submodular
width (a condition strictly more general that bounded fractional
hypertree width) exactly characterizes the fixed-parameter
tractability of the problem. However, the exact condition
characterizing polynomial-time solvability is still an open
problem. In light of the results of \cite{marx-stoc2010-csp}, one
needs to understand the complexity of the problem for classes that
have bounded submodular width, but unbounded fractional hypertree
width.

\bibliographystyle{acmsmall} \bibliography{fractional}

\end{document}